\documentclass{llncs}
\usepackage[utf8]{inputenc}
\usepackage[english]{babel}

\usepackage{cite}
\usepackage{dsfont}
\usepackage{amssymb,amsmath,amsfonts,hyperref}
\usepackage[capitalise]{cleveref}
\usepackage{complexity}
\newclass{\TWOEXP}{2\!-\!EXPTIME}
\newclass{\THREEEXP}{3\!-\!EXPTIME}
\newclass{\EXPTIME}{EXPTIME}
\usepackage{tikz}
\usetikzlibrary {arrows.meta}
\usetikzlibrary{arrows,trees,backgrounds,automata,shapes,plotmarks,decorations}
\tikzstyle{block}=[rectangle,draw, thin, inner sep=3pt, text centered,fill=orange!20!yellow!20] 
\tikzstyle{net}=[draw,cloud,fill=yellow!20,aspect=3,inner sep=1pt]
\tikzstyle{dev}=[draw,circle,fill=yellow!20,aspect=2,inner sep=1pt,minimum size=.6cm]
\tikzstyle{pre}=[<-,shorten <=1pt,>=stealth']
\tikzstyle{post}=[->,shorten >=1pt,>=stealth']
\tikzstyle{bi}=[<->,shorten >=1pt,shorten <=1pt, >=stealth']
\tikzstyle{every initial by arrow}=[initial text={},initial distance=1em,post]
\tikzstyle{every state}=[minimum size=0.6cm,fill=cyan!20!yellow!20]
\tikzstyle{transition}= [post,shorten >=1pt,node distance=2cm, inner sep=2pt,bend angle=20]

\spnewtheorem{con}{Construction}{\bfseries}{\itshape}

\Crefname{theorem}{Thm.}{Thms.}
\Crefname{con}{Constr.}{Constrs.}
\Crefname{proposition}{Prop.}{Props.}
\Crefname{lemma}{Lem.}{Lemmas}
\Crefname{definition}{Def.}{Defs.}
\Crefname{corollary}{Cor.}{Cors.}
\Crefname{figure}{Fig.}{Figs.}
\Crefname{section}{Sect.}{Sections}
\Crefname{equation}{Eq.}{Eqs.}

\newcommand{\real}{\mathds{R}}
\newcommand{\rat}{\mathds{Q}}
\newcommand{\nat}{\mathds{N}}

\newcommand{\NoZ}{\mathop{\nu}}
\newcommand{\aut}{\mathcal{A}}

\newcommand{\gra}{\mathcal{G}}
\newcommand{\ent}{\mathcal{H}}
\newcommand{\capa}{\mathcal{C}}

\newcommand{\bandh}{\mathcal{BH}}
\newcommand{\bandc}{\mathcal{BC}}
\newcommand{\trans}[1]{\xrightarrow{#1}}

\newcommand{\face}{\mathfrak{f}}

\newcommand{\guard}{\mathfrak{g}}
\newcommand{\reset}{\mathfrak{r}}
\newcommand{\card}[1]{\#\!#1}

\newcommand{\0}{\mathbf{0}}
\newcommand{\1}{\mathbf{1}}

\newcommand{\bla}{\boldsymbol{\lambda}}

\newcommand{\dr}{\overrightarrow{d}}

\newcommand{\Path}{\mathop{path}}

\newcommand{\Word}{\mathop{word}}

\newcommand{\length}{\mathop{\mathrm{dur}}}

\newfunc{\myexp}{exp}
\newlang{\myreach}{Reach}
\newlang{\Obese}{Obese}
\newlang{\Meager}{Meager}
\newlang{\Normal}{Normal}

\newcommand{\aaa}{node[blue]{$\bullet$}}
\newcommand{\bbb}{node[red]{$\circ$}}

\newcommand{\DTA}{\textsc{DTA}}
\newcommand{\STA}{\textsc{STG}}
\newcommand{\RTA}{\textsc{RsTA}}

\newcommand{\eugene}[1]{\ifdefined\showcomments \textcolor{olive}{Eugene : #1}\fi}

\title{Computing the Bandwidth of \\ 
Meager Timed Automata\thanks{This work was supported by the ANR project MAVeriQ ANR-CE25-0012 and by the ANR-JST project CyPhAI.}}
\subtitle{(long version)}

\author{Eugene Asarin\inst{1}\orcidID{0000-0001-7983-2202}
\and Aldric Degorre\inst{1}\orcidID{0000-0003-2712-4954}
\and C\u at\u alin Dima\inst{2}\orcidID{0000-0001-5981-4533}
\and \\Bernardo {Jacobo Inclán}\inst{1}\orcidID{0009-0009-5323-7945}}

\institute{Université Paris Cité, CNRS, IRIF, Paris, France
\and LACL, Université Paris-Est Créteil, France}

\authorrunning{Asarin, Degorre, Dima, and Jacobo Inclán}

\begin{document}
\maketitle

\begin{abstract}
The bandwidth of timed automata characterizes the quantity of information produced/transmitted per time unit. We previously delimited 3 classes of TA according to the nature of their asymptotic bandwidth: meager, normal, and obese.
In this paper, we propose a method, based on a finite-state simply-timed abstraction, to compute the actual value of the bandwidth of meager automata.
The states of this abstraction correspond to barycenters of the faces of the simplices in the region automaton. 
Then the bandwidth is  $\log 1/|z_0|$ where $z_0$ is the smallest root (in modulus) of the characteristic polynomial of this finite-state abstraction. 

\end{abstract}
\section{Introduction}\label{sec:intro}
We study timed automata \cite{AD} from the information-theoretic standpoint. An important characteristic of a 
timed automaton is its bandwidth  \cite{formats2022}, a notion generalizing the growth rate of formal languages \cite{ChomskyMiller,growth1,growth2,grigorchuk}. 
In the case of timed automata, bandwidth describes the quantity of information that it can produce/transmit per time unit;
however, differently from the case of discrete formal languages, 
information that timed automata may produce must be related to a degree of precision of an observer of the timed words generated by the timed automaton. 

In the preceding work \cite{3classes} we made a first step in this direction by identifying three classes of TA according to the nature of their asymptotic bandwidth: meager, normal, and obese.
This analysis relies on a refinement of orbit graphs \cite{puri,entroJourn} associated with cycles in the timed automata, which themselves refine the region construction for timed automata.
The three types of timed languages behave very differently. First, meager automata are almost discrete automata, or automata in which real-time delays are very constrained, leading to 
quasi-Zeno behaviors (some quantities decrease while remaining positive, as the distance between Achilles and the Tortoise, see \cite[VI:9]{aristotle} or the more recent paper \cite{BerardPDG98}).
Their bandwidth is bounded by a constant, as is the case for languages of finite automata, hence being independent of the precision of an observer. 
Second, in normal timed languages, information can be encoded every few time units and their bandwidth is $O(\log(1/\varepsilon))$, where $\varepsilon$ is the degree of precision by which an observer
may distinguish timed words generated by the automaton. Thirdly, obese timed languages allow encoding information with a very high frequency, leading to 
$O(1/\varepsilon)$ bandwidth. 

In this paper, we make the first step towards computing the bandwidth of a timed automaton, by attacking the case of meager automata.
Our solution is based on another finite-state abstraction of orbit graphs and region automata, in which each state is associated with a barycenter of one of the faces of the regions. 
The bandwidth can then be computed as $\log 1/|z_0|$ where $z_0$ is the root of the smallest modulus of the characteristic polynomial 
of the adjacency $\lambda$-matrix \cite{lambda} associated 
with the finite-state abstraction. 
The correctness of our construction heavily depends on investigations in \cite[full version]{3classes}, 
which rely on Simon's factorization theorem \cite{simon} and Puri's reachability theory \cite{puri}.


The paper is structured as follows: in \cref{sec:prelims} we recall the main notions and results concerning the bandwidth of timed automata; in \cref{sec:simply} we show how to compute the bandwidth for the restricted class of simply-timed graphs; in \cref{sec:meager} we present our main contribution: a method for computing the bandwidth of meager timed automata (by reduction to simply-timed ones). Most of the proofs can be found in the Appendix. 
\section{Preliminaries}\label{sec:prelims}
\subsection{Timed words, languages and automata}
Timed automata have been introduced in \cite{AD} for modeling and verification of real-time systems.  We will use the main definitions in the following form:
\begin{definition}
    Given $\Sigma$, a finite alphabet of discrete events, a \emph{timed word} over $\Sigma$ is an element from $\left(\Sigma\times\real_+\right)^*$ of the form $w = (a_1,t_1)\dots (a_n,t_n)$, 
with $0 \leq t_1 \leq t_2 \cdots \leq t_n$. We denote $\length(w) = t_n$.
A \emph{timed language} over $\Sigma$ is a set of timed words over the same alphabet.
\end{definition}
For a set of variables $\Xi$, let $G_\Xi$ be the set of finite conjunctions of constraints of the form $\xi \sim b$ with $\xi \in \Xi$, $\sim \in \{<,\leq, >, \geq\}$ and $b \in \nat$.
\begin{definition}
 A  \emph{timed automaton} is a tuple 
$(Q, X, \Sigma, \Delta, S, I, F)$, where
\begin{itemize}
    \item $Q$ is the finite set of discrete locations;
    \item $X$ is the finite set of clocks;
    \item $\Sigma$ is a finite alphabet;
    \item $S,I,F: Q \rightarrow G_{X}$  define respectively the starting, initial, and final clock values for each location;
    \item 
    $\Delta \subseteq Q \times Q \times \Sigma \times G_{X} \times 2^X$ is the transition relation, whose elements are called edges.
\end{itemize}
A timed automaton is \emph{deterministic} (referred to as \DTA) if $\{ (q,x) | x\models I(q)\}$ is a singleton and for any two edges 
$(q,q_1,a,\guard_1,\reset_1)$ and $(q,q_2,a,\guard_2,\reset_2)$ with $q_1 \neq q_2$, the constraint $\guard_1 \wedge \guard_2$ is   non-satisfiable. 
\end{definition}
A small particularity of our definition is the starting clock constraint $S(q)$ defining with which clock values the run can enter each state $q$.

\begin{figure}[t]
\begin{tikzpicture}
%
\node at (0,6) {$\aut_1:$};
\node[state, initial, accepting] at (1,6) (q1) {$q$};
\draw (q1) edge[loop below] node{$a,2<x<3,\{x\}$} (q1);
\node at (2.5,6) {$\aut_2:$};
\node[state, initial, accepting] at (3.5,6) (q1) {$q$};
\draw (q1) edge[loop below] node{$a,b,\,x<1$} (q1);
\draw (q1) edge[loop above] node[right]{$c,5<x<6,\{x\}$} (q1);
\node at (5.5,6) {$\aut_3:$};
\node[state, initial, accepting] at (6.5,6) (q1) {$q$};
\draw (q1) edge[loop below] node{$a,b,\,x<5$} (q1);
\node at (8,6) {$\aut_4:$};
\node[state, initial] at (9,6) (q1) {$q$};
\node[state, accepting] at (11,6) (q2) {$p$};
\draw (q1) edge[->,above, bend left] node{$a,b,\,3 <x<4$} (q2)
(q2) edge[->,below, bend left] node{$\delta_2:b, 5<y<6, \{x,y\}$}(q1);
%
\node at (0,3.5) {$\aut_5:$};
\node[state, initial, accepting] at (1,3.5) (q1) {$q$};
\node[state, accepting] at (3,3.5) (q2) {$p$};
\draw (q1) edge[->,above, bend left] node{$a,b,\,x=3$} (q2)
(q2) edge[->,below, bend left] node{$b,\, y=5, \{x,y\}$} (q1);
\node at (4,3.5) {$\aut_6:$};
\node[state, initial, accepting] at (5,3.5) (q1) {$q$};
\node[state] at (7,3.5) (q2) {$p$};
\draw (q1) edge[->, above, bend left] node{$a,b, x<1, \{x\}$} (q2)
(q2) edge[->, below, bend left] node{$a, b, 1<y<2, \{y\}$} (q1);
\node at (8,3.5) {$\aut_7:$};
\node[state, initial, accepting] at (9,3.5) (q1) {$q$};
\node[state, accepting] at (11,3.5) (q2) {$p$};
\draw (q1) edge[->, above, bend left] node{$a, x<1, \{x\}$} (q2)
(q2) edge[->, below, bend left] node{$b, y<1, \{y\}$} (q1);
%
\end{tikzpicture}

\caption{7 running examples of timed automata, $I(q)=\{\0\}$, $F(\cdot)=\mathbf{true}$ for marked locations. An arrow labeled with $a,b$ is a shorthand for two transitions with the same guards and resets.} \label{fig:5examples}
\end{figure}
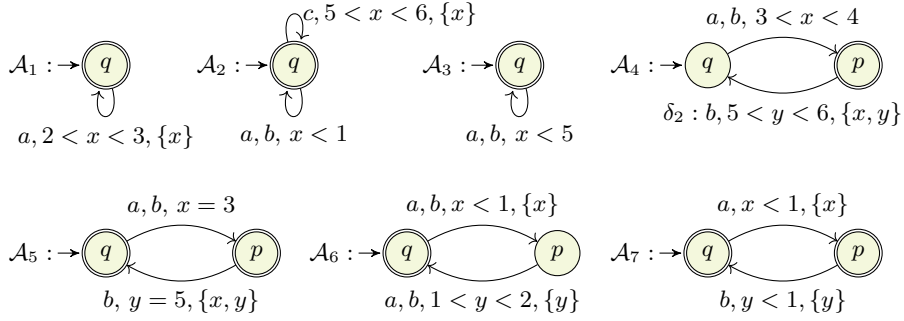


The semantics of a timed automaton is given by a \emph{timed transition system}  whose states are 
tuples $(q,x)$ composed of a location $q\in Q$ and a clock valuation (vector) $x \in [0,\infty)^X$.
Each edge $\delta  = (q,q', a,\guard, \reset) \in \Delta$ generates many timed transitions, which are tuples $(q,x) \trans{\delta,d} (q', x') $ where 
$ x \models S(q)$, $x+d \models \guard$ and $x'_c=0$ whenever $c\in\reset$ and $x'_c=x_c+d$ otherwise, provided that $x' \models S(q')$. We also denote $x[\reset]$ the operation of resetting 
the clocks in $\reset$, hence $x' = (x+d)[\reset]$. 

\emph{Paths} are sequences of edges that agree on intermediary locations. 
At the semantic level, they generate \emph{runs} of the form $$\rho = (q_0,x_0) \trans{(\delta_1,d_1)} (q_1,x_1)\dots \trans{(\delta_n, d_n)}(q_n,x_n),$$ which are sequences of timed transitions that agree on intermediary configurations.
An \emph{accepting run} is a run
in which the first state satisifes $x_0 \models I(q_0)$ and the last one $x_n \models F(q_n)$.

Given a run $\rho$ as above, the path associated with $\rho$ is defined as 
$\Path(\rho)\triangleq\delta_1\dots \delta_n$. 
Furthermore, if $\delta_i = (q_i,q_{i+1},a_i,\guard_i,\reset_i)$, then 
the timed word associated with $\rho$ is defined as $\Word(\rho)\triangleq 
(a_1,d_1)(a_2,d_1+d_2)\ldots  (a_n,\sum_{i\leq n}d_i)$. 
The \emph{duration} of $\rho$ is $\length(\rho) = \length(\Word(\rho))$.
The language of a timed automaton is the set of timed words associated with some accepting run and is denoted $L(\aut)$.
We will also need \emph{time-restricted languages}: for any timed language $L$ and nonnegative real $T \in [0,\infty)$,
$L_T = \{ w \in L \mid \length(w)  \leq T\}$.

For a small example consider a 3-edge path $q\to p\to q$ in the timed automaton $\aut_6$ in \cref{fig:5examples}. One of its runs is 
$$(q,0,0) \trans{(a,0.8)} (p,0,0.8) \trans{(b, 0.7)}(q,0.7,0)\trans{(b,0.2)} (p,0,0.2)$$
and the corresponding timed word is $(a,0.8) (b,1.5) (b,1.7)$.

\paragraph{Region-split form.}
Many algorithms related to timed automata utilize the finitary abstraction of the timed transition system called the \emph{region construction} \cite{AD} that we briefly recall here.
Let us denote by $M$ the maximal constant appearing in the constraints used in the automaton. 
Also, for $\xi \in \real$, we denote $\lfloor \xi \rfloor$ its integral part, 
$\{\xi\}$  its fractional part and, for a clock vector $x\in [0,\infty)^X$, $x_\text{int}\triangleq\{c\in X|x_c\in[0,M]\cap\nat\}$, and $x_\text{frac}\triangleq\{c\in X| x_c \in [0,M]\setminus\nat\}$.
Then regions are defined as equivalence classes of the region equivalence described hereafter. Two clock vectors $x$ and $y$ are region-equivalent iff~:
%
\begin{itemize}
    \item $x_\text{int}=y_\text{int}$ and $x_\text{frac} = y_\text{frac}$;
    \item and for any clock $ c\in x_\text{int}\cup x_\text{frac}, \lfloor x_c \rfloor = \lfloor y_c \rfloor$;
    \item and for any two clocks  $c_1, c_2\in x_\text{int}\cup x_\text{frac}$, $\{x_{c_1}\}\leq \{x_{c_2}\}$ iff $\{y_{c_1}\} \leq \{y_{c_2}\}$.
\end{itemize}
Note that a region $R$ is a simplex of some dimension $d\leq \#X$. 


As in \cite{entroJourn,3classes} we utilize here timed automata in which the guards defining starting clock values for each location actually define regions:
\begin{definition}
A \emph{region-split 
TA} (or \RTA\ for short) is a deterministic timed automaton $(Q,X,\Sigma,\Delta,S,I,F)$, such that, for any location $q\in Q$: 
\begin{itemize}
\item $S(q)$ defines a non-empty bounded region, called the \emph{starting region} of $q$;
\item all states in $\{q\}\times S(q)$ are reachable from the initial state and co-reachable to a final state;
\item either $I(q)=S(q)$ is a singleton\footnote{by definition of DTA this is possible for a unique location $q$} or $I(q)=\emptyset$;
\item for any edge $(q,q',a,\guard,\reset)\in\Delta$,  
$\left(\{S(q)+d\mid d\in \real_+\}\cap \guard\right)[\reset]=S(q')$, where we utilized the $(\cdot) [\reset]$ operator lifted to sets of clock valuations.
\end{itemize}
\end{definition}

A folklore result says that any \DTA\ can be brought into a region-split form which may contain an exponentially larger set of locations.
Note that this construction is similar to that of the region automaton introduced in \cite{AD}, except that the outcome is slightly coarser and is typed as a timed automaton. As an example, we present on \cref{fig:meager-example} the region-split form $\aut_6'$ of the automaton $\aut_6$ (from \cref{fig:5examples}) with explicit starting regions $S(\cdot)$  for each state. In particular, state $q$ of $\aut_6$ is split into two states: the initial $s$ and $q$ corresponding to zero and non-zero clock vectors, respectively. 


We also introduce the notation $\bar\aut$ for the closed version of any \RTA\ $\aut$: i.e. the same automaton where all strict inequalities are changed to non-strict ones. Thanks to region-splitting, the closure retains the same set of symbolic behaviors as the original \RTA, this allows us to safely reason on runs that visit vertices and faces of regions.

\subsection{Pseudo-distance on timed words}
We are using the proximity measure on timed words introduced in \cite{distance}.
\begin{definition}
The \emph{pseudo-distance} $d$ between timed words 
$$
w = (a_1,t_1)\dots(a_n,t_n) \text { and }v = (b_1,s_1)\dots(b_m,s_m)$$ 
is defined as follows (with the convention $\min \emptyset = \infty$):
$$
 \dr(w,v)\triangleq \max_{i\in\{1..n\}}\min_{j\in\{1..m\}} \{ |t_i-s_j|:a_i=b_j \};\hfill
 d(w,v) \triangleq  \max( \dr(w,v) , \dr(v,w) ).
$$
\end{definition}
This pseudo-distance allows a meaningful comparison of timed words with a different number of events. It is illustrated on \cref{fig:dist}.
Intuitively, two words are close to each other when they cannot be distinguished by an observer that reads the discrete letters of the word exactly 
(they can determine whether or not a letter has occurred) but with some imprecision w.r.t.~time. 
So the observer cannot determine when two letters are very close to one another, which one came before the other, and not even how many times a letter was repeated within a short interval. 
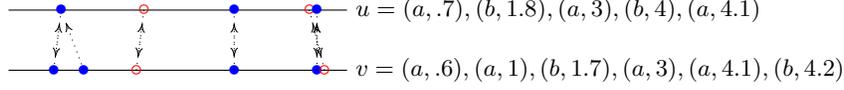
\begin{figure}[t]
\begin{center}
{
\usetikzlibrary {arrows.meta,positioning} 
\begin{tikzpicture}
\draw (0,.8) -- (.7,.8) \aaa  --(1.8,.8) \bbb  -- (3,.8) \aaa-- (4,.8) \bbb-- (4.1,.8) \aaa --
(4.5,.8)node[anchor=west]{$u=(a,.7),(b,1.8), (a,3), (b,4), (a,4.1)$};

\draw (0,0) -- (.6,0) \aaa   -- (1,0) \aaa--(1.7,0) \bbb  -- (3,0) \aaa-- (4.1,0) \aaa-- (4.2,0)\bbb --
(4.5,0)node[anchor=west]{$v=(a,.6),(a,1), (b,1.7), (a,3), (a,4.1),(b,4.2)$};

\draw [dotted, arrows={
->[width=1mm,length=1.2mm,sep=1.8mm]}]
(.6,0) edge (.7,.8)
(.7,.8) edge (.6,0)
(1,0) edge (.7,.8)
(1.7,0) edge (1.8,.8)
(1.8,.8) edge (1.7,0)
(3,0) edge (3,.8)
(4,.8) edge (4.2,0)
(3,.8) edge (3,0)
(4.2,0) edge (4,.8)
(4.1,0) edge (4.1,.8)
(4.1,.8) edge (4.1,0)
;
\end{tikzpicture}
}
\end{center}
\caption{Pseudo-distance between two timed words (dotted lines with arrows represent directed matches between letters).  	$\protect\dr(u,v)=0.2;\  \protect\dr(v,u)=0.3$, thus $ d(u,v)=0.3$.}\label{fig:dist}
\end{figure}

It is possible that $d$ fails to distinguish timed words, that is $d(w_1,w_2)=0$  but $w_1\neq w_2$; this could happen when $w_1$ and $w_2$
only differ by order and quantity of simultaneous letters, e.g.~for $w=(a,1)(b,1)$ and $v=(b,1)(b,1)(a,1)$, that is why $d$ is  only a pseudo-distance.

We will sometimes use the following trick to ``transform'' $d$ into a distance. We call a timed word $w = (a_1, t_1)\dots (a_n, t_n)$ \emph{0-free} whenever $0<t_1<t_2<\cdots <t_n$. 
Note that $d$ satisfies the axiom $w_1=w_2\Leftrightarrow d(w_1,w_2)=0$ whenever $w_1,w_2$ are 0-free, hence $d$ becomes a distance on 0-free words.
\begin{definition}[0-elimination, words]\label{def:nu} Given a timed word $w$ over an alphabet $\Sigma$, we define its \emph{$0$-elimination} as the 0-free timed word $W=\NoZ(w)$ over the alphabet $2^\Sigma\setminus\{\emptyset\}$ as follows:
\begin{itemize}
    \item let $0<t_1<t_2<\cdots<t_n$ be all the distinct non-zero event times in $w$;
    \item let $A_i\in 2^\Sigma$ be the set of all the events in $w$ occurring at time $t_i$;
    \item finally let $W=(A_1,t_1),(A_2,t_2),\dots, (A_n,t_n)$.
\end{itemize}
\end{definition}
For example $\NoZ((c,0)(a,5)(b,5)(a,5)(c,7))=(\{a,b\},5)(\{c\},7)$, the events at time $0$ are ignored.
We summarize the properties of the operation $\NoZ$:
\begin{proposition}\label{eq:nu}
Operation $\NoZ$ maps  timed words on $\Sigma$ to $0$-free ones on $2^\Sigma\setminus\{\emptyset\}$.
$\NoZ(w_1)=\NoZ(w_2)$ if and only if the words $w_i, i=1,2$ admit decompositions $w_i=u_iv_i$ with $\length(u_i)=0$ and $d(v_1,v_2)=0$.
\end{proposition}
In words, 0-eliminations of two timed words coincide if and only if the words are at distance $0$ except for the initial time instant $0$.

\subsection{Bandwidth of timed languages}

The definition of bandwidth is based on Kolmogorov\&Tikhomirov's work \cite{kolmoEpsilon} 
formalizing the quantity of information in an element of a compact metric space observed with a given precision. 
We recall the adaptation of their definitions from \cite{distance} to the context of timed languages.
\begin{definition}
Let $L$ be a timed language, then:
\begin{itemize}
    \item ${\mathcal M} \subseteq L$ is an $\varepsilon$\emph{-separated subset} of $L$ if $\forall x\neq y \in{\mathcal M} , d(x,y) > \varepsilon$;
    \item ${\mathcal N}$ is an $\varepsilon$\emph{-net} of $L$  if $\forall y \in L, \exists s \in {\mathcal N} \text{ s.t. } d(y,s) \leq \varepsilon$;
    \item $\varepsilon$\emph{-capacity} of $L$ is defined\footnote{all logarithms in this article are base 2} as 
    $$
     \capa_\varepsilon(L) \triangleq  \log\max\{\card{\mathcal M}\mid  {\mathcal M}  \text{ $\varepsilon$-separated set of } L\};
    $$
    \item
    $\varepsilon$\emph{-entropy}  of $L$ is defined as
    $$
    \ent_\varepsilon(L) \triangleq  \log\min\{\card{\mathcal N}  \mid  {\mathcal N}   \text{ $\varepsilon$-net of $L$}\}.
    $$
\end{itemize}
\end{definition}
As shown in \cite{distance}, every time-bounded timed language is compact with respect to our distance $d$ and thus has finite  $\varepsilon${-capacity} and $\varepsilon${-entropy}. 
 \cite {formats2022} introduced the notion of bandwidth of a timed language as the entropy or capacity per time unit and related this notion to bounded delay codes.

\begin{definition}[bandwidth, \cite{formats2022}]The $\varepsilon$\emph{-entropic} and $\varepsilon$\emph{-capacitive bandwidths} of a timed  language  $L$ are 
\[
\bandh_\varepsilon(L) \triangleq \limsup_{T\rightarrow\infty}\frac{\ent_\varepsilon(L_T)}{T}; \qquad
\bandc_\varepsilon(L) \triangleq  \limsup_{T\rightarrow\infty}\frac{\capa_\varepsilon(L_T)}{T}. 
\]
\end{definition}
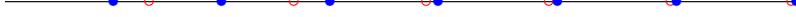
\begin{figure}[t]
%
    \centering
    \begin{tikzpicture}[scale=1.6]
		\draw (0,0) -- (.9,0) \aaa   -- (1.2,0) \bbb--(1.8,0) \aaa  -- (2.4,0)\bbb -- (2.7,0)\aaa-- (3.5,0) \bbb-- (3.6,0)\aaa-- (4.52,0) \bbb-- (4.59,0)\aaa-- (5.53,0) \bbb-- (5.58,0)\aaa-- (6.535,0) \bbb-- (6.575,0)\aaa
  ;
  \end{tikzpicture}
\caption{
A run of $\aut_6$ on $(a,.9)(b,1.2)(a,1.8)(b,2.4)(a,2.7)(b,3.5)(a,3.6)(b,4.52)\\(a,4.59)(b,5.53)(a,5.58)(b,6.535)(a,6.575)$.} \label{fig:arun}
\end{figure}

\begin{figure}[t]  
 \begin{tikzpicture}
 \node at (1.5,2) {$\aut_6':$};
 \node[state, ellipse, initial, accepting] at (3,1) (q0) {$\displaystyle\frac{s}{x=y=0}$};
  \node[state,ellipse,  accepting] at (6,1) (q1) {$\displaystyle\frac{q}{0=y<x<1}$};
\node[state, ellipse] at (11,1) (q2) {$\displaystyle\frac{p}{0=x<y<1}$};
\draw (q0) edge[->, above, bend left] node{$a,b, x<1, \{x\}$} (q2)
 (q1) edge[->, above] node{$a,b, x<1, \{x\}$} (q2)
(q2) edge[->, below, bend left] node{$a,b, 1<y<2, \{y\}$} (q1);
\end{tikzpicture}

\caption{
The region-split form of $\aut_6$, constraints within states correspond to start conditions $S(\cdot)$.
} \label{fig:meager-example}
\end{figure}
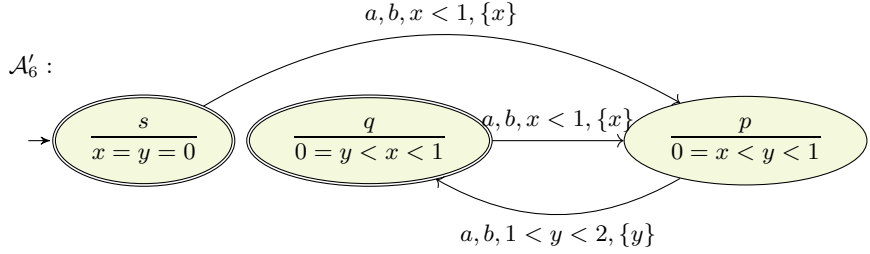

\begin{definition}[3 classes, \cite{3classes}]
A timed language $L$ is 
\begin{itemize}
    \item \emph{meager}
    whenever $\bandh_\varepsilon(L) = O(1)$;
    \item \emph{normal} whenever $\bandh_\varepsilon(L) =  \Theta(\log \frac{1}{\varepsilon})$;
    \item \emph{obese} whenever $\bandh_\varepsilon(L) =  \Theta(\frac{1}{\varepsilon})$, as $\varepsilon\to 0$.
\end{itemize}

\end{definition}
As the main result of \cite{3classes} we proved that every timed regular language belongs to one of these classes, and provided a classification algorithm for timed automata. 


\cref{fig:5examples} gives a few examples of automata in the three classes: $\aut_3$, $\aut_5$ and $\aut_6$ are meager,  $\aut_1$ and $\aut_4$ are normal and $\aut_2$ and $\aut_7$ are obese. From now on, we concentrate on meager timed automata and compute their bandwidth. 
$\aut_5$ has only discrete choices since its transitions happen at discrete dates only. $\aut_3$ can produce a huge amount of information ($2/\varepsilon$ bits every second), but only during the first five seconds of its life. Its bandwidth is $\lim_{T\to \infty} 2/\varepsilon T=0$. The automaton $\aut_6$  (serving as running example below) is much less evident. It involves two flows of events, with events of the former being $<1$ second apart and of the latter on the contrary $>1$ second apart, as shown on \cref{fig:arun}. 
Such interleaving is possible for any duration $T$, but real-valued delays between symbols become more and more constrained.

\section{Simply-timed graphs and their bandwidth}\label{sec:simply}
Following \cite{realtime,simply}, we define a class of graphs that will serve as the main abstraction for computing the bandwidth of meager timed automata.

\begin{definition}[Simply-timed graphs] A \emph{simply-timed graph} (\STA) is a tuple $\gra=(Q,\Sigma,\Delta)$ with $Q$ a finite set of states, $\Sigma$ a finite alphabet, and $\Delta\subset Q\times\rat_+\times\Sigma\times Q$ a finite transition relation. For a transition $(p,d,a,q)\in\Delta$, the state $p$ is referred to as  origin, $q$ destination, $a$ label and $d$ delay.

An \STA\ is \emph{$0$-free} if all its delays are strictly positive. 
It is \emph{deterministic} if for any two transitions $(p,d,a,q)$ and $(p,d,a,q')$ necessarily $q=q'$. 
\end{definition}

The timed language of such a graph is defined in a natural way.
\begin{definition}[Semantics of \STA] A \emph{run} of an \STA\ $(Q,\Sigma,\Delta)$ on a timed word $w= (a_1,t_1)\dots(a_n,t_n)$ is a sequence
$
q^0 \trans{d_0,a_1} q^1  \trans{d_1,a_2}\cdots q^n,
$
such that
$ t_i-t_{i-1}=d_i$, and $(q^{i-1},d_i,a_i,q^i)\in\Delta$.
The \emph{language} $L(\gra)$ of an \STA\ $\gra$ is the set of all timed words having a run.
\end{definition}
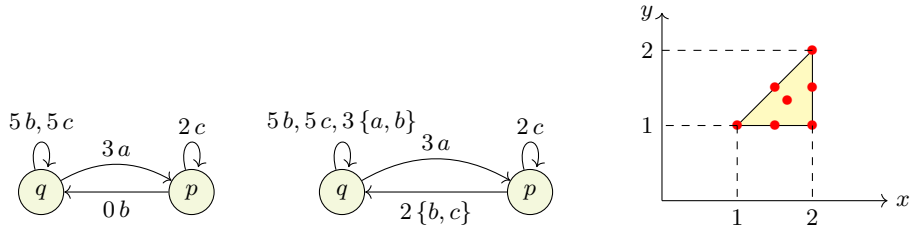
\begin{figure}[t]
\begin{tikzpicture}
\node[state] at (1,1) (q) {$q$};
\node[state] at (3,1) (p) {$p$};
\draw (q) edge[->, above,bend left] node{$3\,a$} (p)
(p) edge[->, below] node{$0\,b$} (q)
(q) edge[loop above] node{$5\,b, 5\,c$} (q)
(p) edge[loop above] node{$2\,c$} (p);
\node[state] at (5,1) (qq) {$q$};
\node[state] at (7.5,1) (pp) {$p$};
\draw (qq) edge[->, above,bend left] node{$3\,a$} (pp)
(pp) edge[->, below] node{$2\,\{b,c\}$} (qq)
(qq) edge[loop above] node{$5\,b, 5\,c, 3\,\{a,b\}$} (qq)
(pp) edge[loop above] node{$2\,c$} (pp);

\end{tikzpicture}\hfill 
\begin{tikzpicture}
\draw[fill=yellow!30] (1,1) -- (2,2) -- (2,1) --cycle;
\draw[color=red] (1,1) node {$\bullet$} (2,2) node {$\bullet$} (2,1) node {$\bullet$}
(1.5,1) node {$\bullet$} (1.5,1.5) node {$\bullet$} 
(2,1.5) node {$\bullet$} (5/3,4/3 )node {$\bullet$};

\draw[->] (0,0) -- (3,0);
\draw (3,0) node[right] {$x$};
\draw [->] (0,0) -- (0,2.5);
\draw (0,2.5) node[left] {$y$};
\draw [dashed] (1,1) -- (1,0) node[below] {$1$};
\draw [dashed] (2,1) -- (2,0) node[below] {$2$};
\draw [dashed] (1,1) -- (0,1) node[left] {$1$};
\draw [dashed] (2,2) -- (0,2) node[left] {$2$};

\end{tikzpicture}
    \caption{ A simply-timed graph (left), and its 0-free form (middle). Right: a 2-dimensional region, its 7 faces, and their barycenters.}
    \label{fig:STA}
\end{figure}

The notions of $\varepsilon$-bandwidth are simpler for \STA. 
\begin{definition}
    For a set $S$ of timed words, we call its \emph{size} and denote by $\Upsilon(S)$ the cardinality of the largest $0$-separated set in $S$.
\end{definition}
Equivalently, this is the cardinality of the smallest 0-net, and the number of equivalence classes for the binary relation $d(u,v)=0$ over $S$. We remark that whenever $S$ is 0-free, $\Upsilon(S)=\# S$.
\begin{definition}
The \emph{growth rate} of an \STA\ $\gra$ is defined as
$$
\beta(\gra)=\lim_{T\to\infty}\frac{\log \Upsilon(L_T(\gra))}{T},
$$
whenever this limit exists and is finite.
\end{definition}

\begin{proposition}\label{prop:STA}Given an \STA\ $\gra$ with non-empty language, the growth rate $\beta(\gra)$ always exists. Moreover, taking $N$ the common denominator of all delays,   for $\varepsilon<1/2N$ the growth rate coincides with both bandwidths:
$$
\bandc_\varepsilon(L(\gra))=\bandh_\varepsilon(L(\gra))=\beta(\gra).
$$
\end{proposition}

We present the method for computing $\beta(\gra)$, which consists of several transformations.
First, we relate the notions of size and growth rate to 0-elimination.
\begin{proposition}\label{eq:nu:L} 
0-elimination preserves the growth rate: $\beta(L)=\beta(\NoZ(L))$.

\end{proposition}

Thanks to this proposition, 
instead of counting equivalence classes in $L_T$ we will count timed words in $\NoZ(L)_T$. To that aim, we will transform the \STA\ for $L$ to that for $\NoZ(L)$.

\begin{con}[0-elimination in \STA]
Given an \STA\ $\gra=(Q,\Sigma,\Delta)$,  we define $\NoZ(\gra)= (Q,2^\Sigma\setminus\{\emptyset\},\Delta')$ 
combining every timed transition with each subsequent instant transition: $(p,d,A,q)\in\Delta'$ whenever $d>0$ and there exist  $(p,d,a_0,q_0),(q_0,0,a_1,q_1),\dots, (q_{k-1},0,a_k,q_k) \in\Delta$ with  $q_k=q$,  $A=\{a_0,\dots, a_k\}$.
\end{con}
We remark that  $\NoZ(\gra)$ can be easily computed by saturation. 
\begin{proposition}
Whenever an \STA\ $\gra$  recognizes $L$, the 0-eliminated graph $\NoZ(\gra)$ recognizes the language  $\NoZ(L)$. 
\end{proposition}

    Given an \STA\ $\gra$, we compute its determinization by the usual subset construction, see e.g.~\cite[Thm 3.3.2]{marcus}.   Determinization preserves the timed language and thus the growth rate.

Given a deterministic $0$-free \STA\ $\gra$, we can compute its growth rate using the technique of generating functions. 

\begin{con}
    Given an \STA\ $\gra=(Q,\Sigma,\Delta)$ with $Q=\{q_1,\dots,q_n\}$, we define its $n\times n$ \emph{adjacency matrix} as 
    $$M_\gra(z)=(m_{ij}(z)) \mathrm{\ with\ } m_{ij}(z) =\sum_{(q_i,d,a,q_j)\in\Delta} z^d,$$
    and its \emph{characteristic quasi-polynomial} as $\psi_\gra(z)=\det(I-M_\gra(z))$.
\end{con}
\begin{proposition}\label{prop:char}
    Let $\gra$ be a  deterministic $0$-free \STA, and let $z_0$ be the smallest root (in modulus) of its \emph{characteristic equation} $\psi_\gra(z)\!=\!0$. Then $\beta(\gra)\!=\!-\!\log|z_0|$.
\end{proposition}

\begin{theorem}\label{thm:algo:STA}
  Given an \STA\  $\gra$, for $\varepsilon$ small enough, $\bandc_\varepsilon ( \gra)=\bandh_\varepsilon( \gra)=\beta(\gra)$. The growth rate  $\beta(\gra)$ is a logarithm of an algebraic number computable as a function of $\gra$.
\end{theorem}
\begin{proof}
The algorithm for computing $\beta(\gra)$ is as follows:  eliminate 0, determinize, write the adjacency matrix, and solve the characteristic equation. Its correctness follows from the above propositions. \qed
\end{proof}

As an example consider the \STA\ on \cref{fig:STA} (left). Its 0-free form is represented in the middle, it is deterministic. The characteristic matrix  and equation are
$$
M(z)=
\begin{pmatrix}
 z^3+2z^5 &    z^3\\
 z^2 & z^2
\end{pmatrix};
\  \det(I-M(z))=2 z^7 - 2 z^5 - z^3 - z^2 + 1=0, 
$$
with the smallest root $z_0\approx 0.698776$, the growth rate is $-\log_2 z_0\approx 0.517098$.

\section{Computing the bandwidth of a meager \RTA}\label{sec:meager}


Consider a closure of a $(c-1)$-dimensional clock region $\bar{r}$  in the clock space of $\aut$. It is a simplex with $c$ vertices. 
The convex hull of any subset of $k>0$ vertices of a region will be referred to as a $(k-1)$-dimensional \emph{face}. 
Note that each vertex and the whole closed region are considered as faces. 
To each such face $f$ with vertices $v_1,\dots,v_k$  we associate its \emph{barycenter} $b=\frac1k\sum_{i=1}^kv_k$, 
as illustrated on \cref{fig:STA}, right. We denote the set of all such face barycenters $\alpha(r)$.

Given a clock vector $x\in\bar{r}$, we consider the smallest face $\face(x)$ containing $x$ and denote $\alpha(x)$ the barycenter of $\face(x)$. 
Abusively, we also say that  $x$ and $\alpha(x)$ have the dimension $k-1$.

\paragraph{Barycentric abstraction.}

Our method for the computation of the bandwidth of an \RTA\ involves abstracting the 
\RTA\ by an \STA\ whose states correspond to barycenters of the faces of each start region in the \RTA. 
The transitions of the \STA\ are exactly the timed transitions of the original \RTA\ that allow a unique time delay 
and do not change face dimension. 
Formally:
\begin{con}\label{con:abstr}
    Given an \RTA\ $\aut=(Q,X,\Sigma,\Delta,S,I,F)$, its \emph{barycentric abstraction} is an \STA\ $\gra=\alpha(\aut)=(Q',\Sigma,\Delta')$ obtained as follows:
    \begin{itemize}
    \item the state space is constituted by face barycenters of all the starting regions:
    $$
    Q'=\bigcup_{q\in Q}\{q\}\times \alpha(\overline{S(q)});
    $$
    \item for each edge $\delta=(p,p',a,\guard,\reset)\in\Delta$, and for each couple of barycenters  $x\in \alpha(S(p))$ and $x'\in \alpha(S(p'))$ of the same dimension,\\
        if $L_{\bar{\delta}}((p,x),(p',x'))=\{ta\}$ for some unique $t\in\rat$, then $\Delta'$ contains a transition $((p,x),t,a,(p',x'))$.
    $\Delta'$ contains no other transitions.
    \end{itemize}
\end{con}
Note that the abstraction leaves some edges of the \RTA\ without a counterpart in the built \STA. In fact, these removed edges can only appear a bounded number of times within a run of a meager automaton, and thus are actually not needed to produce the full amount of bandwidth.

For our running example,  the region-split automaton $\aut_6'$ on \cref{fig:meager-example} the starting region is 0-dimensional for $s$, just $(0,0)$. The two other starting regions are 1-dimensional: $(0,1)\times\{0\}$ for $q$ and $\{0\}\times(0,1)$ for $p$, in each of those there are two vertices and one 1-dimensional barycenter in the middle.  The barycentric abstraction is presented on \cref{fig:B6}, top; it splits into two connected components: one on vertices (left subfigure) and one on region centers $(1/2,0)$ for $q$  and $(0,1/2)$ for $p$ (right).


%
The theorem below states that the bandwidth of the apparently smaller 
language $ L(\gra)$ is indeed the same as that of  $L(\bar\aut)$, implying that the barycentric abstraction is sufficient for computing the bandwidth of meager automata.




\begin{theorem}\label{prop:main}
    For any meager \RTA\ $\aut$ and its barycentric abstraction $\gra=\alpha(\aut)$, and for $\varepsilon$ small enough it holds that
    $$
    \bandc_\varepsilon(\aut)=\bandh_\varepsilon(\aut)=\beta(L(\gra)).
    $$
\end{theorem}
Let us sketch the main ideas of the proof. 

The \textbf{lower bound} is not difficult. By the pigeonhole principle, there exists a  node $n=(p,x)$ of $\gra$, such $L_\gra(n,n)$ has growth rate $\beta(L(\gra))$. Let now $u$ a timed word leading in $\bar\aut$ from $I$ to $(p,x)$  and $v$  leading from $(p,x)$ to $F$. Then it is easy to see (directly from \cref{con:abstr}) that $u\cdot L_\gra(n,n)\cdot v\subset L(\bar\aut)$ and hence $\beta(L(\gra))=\bandh_0(u\cdot L_\gra(n,n)\cdot v)\leq \bandh_0(u\cdot L_\gra(n,n)\cdot v)\leq \bandh_0(L(\bar\aut)) = \bandh_0(L(\aut))$.\footnote{The latter equality is a consequence of \cite{3classes}, full version, Lem.~35.}

For the \textbf{upper bound}, as established in \cite{3classes}, the totality of the bandwidth of meager automata is produced by cycles running from a unique state $(q,x)$ (in the closure of the automaton) to itself. Moreover, each path from $(q,x)$  to itself admits a unique timing. We prove that this bunch of cycles can be approximated by a similar bunch on the barycenter  $(q,\alpha(x))$. We count elements of the latter using the \STA\ $\gra$. 

As a comment, while our reasoning could possibly still be valid for other choices of face representatives than barycenters (any interior point of a face can represent it), our choice is not arbitrary because it allows for a terser abstraction, based on a single, canonical, representative per face. Indeed, when an edge in a cycle of a \RTA\ sends a face onto another face, the barycenter of the destination face is always a successor of the barycenter of the origin face.

Our main result is immediate from  \cref{prop:main,thm:algo:STA}.
\begin{theorem}\label{thm:main}
   The bandwidth of a meager \DTA\   $\aut$, for $\varepsilon$ small enough, satisfies  $\bandc_\varepsilon ( \aut)=\bandh_\varepsilon( \aut)=\beta$ with $\beta$   a logarithm of an algebraic number computable as a function of $\aut$.
\end{theorem}
We remark that the overall complexity of our algorithm is doubly exponential: one exponent is due to region-splitting and taking barycenters of all faces, the other to determinizing the \STA. This is to compare to  \PSPACE-completeness of recognizing meager \DTA \cite{3classes}.

\begin{figure}[t]
\centering
\begin{tikzpicture}\scriptsize
\node[circle, inner sep=0.5mm,fill=blue,left] (s) at (-2,0) {};
\node[black] at (-2.3,0) {$s$};
  \node[circle, inner sep=0.5mm,fill=blue] (a) at (0,0) {};
  \node[circle, inner sep=0.5mm,fill=blue]  (b) at (1,0) {};
  \node[circle,inner sep=0.5mm,fill=blue] (c) at (4,0) {};
  \node[circle,fill=blue,inner sep=0.5mm] (d) at (4,1) {};  
  \path[-,dashed,blue] 
  (a) edge node[below,black,inner sep=0.5mm]{$q$}(b) 
  (c) edge node[right,black,inner sep=0.5mm]{$p$}(d);
 \path[-{Latex},  thin,black!50!green] 
  (s) edge[bend right] node[below]{$0,a,b$} (c)
 (s) edge[bend left] node[above]{$1,a,b$}(d)
 (a) edge[bend right] node[below, near start]{$0,a,b$} (c)
 (a) edge[bend left] node[above]{$1,a,b$}(d)
 (b) edge node[above right]{$0,a,b$}(c)
 
  (c) edge[bend left]node[above]{$1,a,b$} (b)
 (d) edge node[ right]{$1,a,b$} (b)
  (d) edge node[above]{$0,a,b$} (a);

  \node (a) at (4.5,0) {};
  \node  (b) at (5.5,0) {};
  \node (c) at (7.5,0) {};
  \node (d) at (7.5,1) {};
    \node[circle, inner sep=0.5mm,fill=blue]  (ab) at (5,0) {};
  \node[circle,inner sep=0.5mm,fill=blue] (cd) at (7.5,0.5) {};
  \path[-,dashed, blue] 
    (a) edge node[below,black]{$q$}(b)
  (c) edge node[right,black]{$p$}(d);
  \path[-{Latex},  thin,black!50!green]  
  (ab) edge [bend left] node[above]{$1/2,a,b$}(cd)
  (cd) edge node[below]{$1/2,a,b$} (ab);
\end{tikzpicture}\\
\begin{tikzpicture}\scriptsize
\node[circle, inner sep=0.5mm,fill=blue]  (s) at (6,0) {};
  \node[circle, inner sep=0.5mm,fill=blue] (a) at (8.5,0) {};
  \node[circle, inner sep=0.5mm,fill=blue]  (b) at (9.5,0) {};
  \node[circle,inner sep=0.5mm,fill=blue] (c) at (11.5,0) {};
  \node[circle,fill=blue,inner sep=0.5mm] (d) at (11.5,1) {};  
 \path[-{Latex},  thin,black!50!red] 
  (s) edge[above ] node[above]{$1,a,b,ab$} (a)
 (s) edge[bend left] node[above]{$1,a,b$}(d)
 (s) edge[bend right] node[below]{$1,a,b,ab$} (c)
 (a) edge[loop above ] node[above]{$1,a,b,ab$} (c)
 (a) edge node[above]{$1,a,b$}(d)
 (a) edge[bend right] node[below, near start]{$1,a,b,ab$} (c)

 (c) edge node[above right]{$1,a,b$} (b)
(c) edge[loop below ] node[below]{$1,a,b,ab$} (c)
 (d) edge node[right]{$1,a,b$} (b)
 (d) edge node[below,rotate=90]{$1,a,b,ab$} (c);

     \node[circle, inner sep=0.5mm,fill=blue]  (ab) at (13.5,0) {};
  \node[circle,inner sep=0.5mm,fill=blue] (cd) at (16,0.5) {};

  \path[-{Latex},  thin,black!50!red]  
  (ab) edge [bend left] node[above]{$1/2,a,b$}(cd)
  (cd) edge node[below]{$1/2,a,b$} (ab);

\end{tikzpicture}
\caption{Top: barycentric abstraction $\gra=\alpha(\aut_6')$ split in  components of dimension 0 and 1. 
Bottom: the same graph after 0-elimination.
Multiple labels on an edge mean that there is an edge for each label, all with the same timing.}\label{fig:B6}
\end{figure}
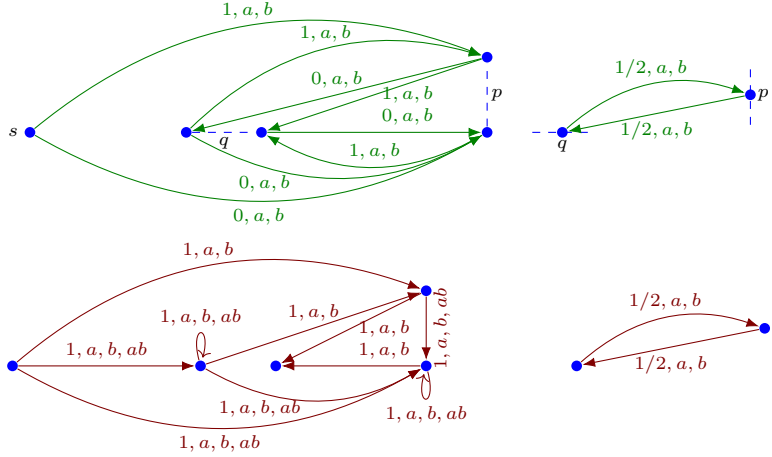

For the running example $\aut_6$, we proceed with 0-elimination for the \STA\ on \cref{fig:B6} (top), the result is given on  \cref{fig:B6} (bottom). The graph obtained is deterministic. The adjacency 
 matrix of the 0-free \STA\ obtained splits into two diagonal blocks (corresponding to connected components of barycenters of dimensions 0 and 1):
$$
\begin{pmatrix}
0 & 3z &0 & 3z& 2z\\
0 & 3z &0 & 3z& 2z\\
0 & 0 & 0 & 0 &0 \\
0 & 0 & 2z & 3z &0\\
0 & 0 &2z &3z &0
\end{pmatrix};
\quad
\begin{pmatrix}
0 & 2z^{1/2}\\
2z^{1/2} & 0
\end{pmatrix}.
$$
The characteristic equation also splits:
$$
(1 - 3 z)^2(1-4z)=0.
$$
The smallest root (corresponding to dimension 1 barycenters) is $1/4$, the bandwidth is $-\log 1/4=2$.

\section{Conclusions}
We have solved the problem of computing the bandwidth for simply-timed graphs and based on that, for all meager automata.
Our main tool is the barycentric abstraction, a new construction that can be seen as a refinement of the corner point automaton from \cite{stayingalive} (which, in turn, refines the region automaton \cite{AD}). 

The next step is to build appropriate abstractions for computing the bandwidth of normal and obese automata, using the tools from \cite{3classes}. Synthesis of timed codes and their applications would be the outcome of this research program.
\bibliographystyle{splncs04}
\bibliography{entro}

\begin{thebibliography}{10}
\providecommand{\url}[1]{\texttt{#1}}
\providecommand{\urlprefix}{URL }
\providecommand{\doi}[1]{https://doi.org/#1}

\bibitem{AD}
Alur, R., Dill, D.L.: A theory of timed automata. Theoretical Computer Science
  \textbf{126},  183--235 (1994). \doi{10.1016/0304-3975(94)90010-8}

\bibitem{aristotle}
Aristotle: Physics (350 BCE),
  \url{https://classics.mit.edu/Aristotle/physics.html}

\bibitem{entroJourn}
Asarin, E., Basset, N., Degorre, A.: Entropy of regular timed languages.
  Information and Computation  \textbf{241},  142--176 (2015).
  \doi{10.1016/j.ic.2015.03.003}

\bibitem{distance}
Asarin, E., Basset, N., Degorre, A.: Distance on timed words and applications.
  In: Proc.~FORMATS. LNCS, vol. 11022, pp. 199--214 (2018).
  \doi{10.1007/978-3-030-00151-3_12}

\bibitem{3classes}
Asarin, E., Degorre, A., Dima, C., Jacobo~Incl\'an, B.: Bandwidth of timed
  automata: 3 classes. In: Proc. {FSTTCS}. LIPIcs, vol.~284, pp. 10:1--10:17
  (2023). \doi{10.4230/LIPICS.FSTTCS.2023.10}, full version
  \url{https://doi.org/10.48550/arXiv.2310.01941}

\bibitem{BerardPDG98}
B{\'{e}}rard, B., Petit, A., Diekert, V., Gastin, P.: Characterization of the
  expressive power of silent transitions in timed automata. Fundamenta
  Informaticae  \textbf{36}(2-3),  145--182 (1998). \doi{10.3233/FI-1998-36233}

\bibitem{stayingalive}
Bouyer, P., Brinksma, E., Larsen, K.G.: Staying alive as cheaply as possible.
  In: Proc.~{HSCC}. LNCS, vol.~2993, pp. 203--218 (2004).
  \doi{10.1007/978-3-540-24743-2_14}

\bibitem{growth1}
Bridson, M.R., Gilman, R.H.: Context-free languages of sub-exponential growth.
  Journal of Computer and System Sciences  \textbf{64}(2),  308--310 (2002).
  \doi{10.1006/jcss.2001.1804}

\bibitem{ChomskyMiller}
Chomsky, N., Miller, G.A.: Finite state languages. Information and Control
  \textbf{1}(2),  91 -- 112 (1958). \doi{10.1016/S0019-9958(58)90082-2}

\bibitem{realtime}
Dima, C.: Real-time automata. Journal of Automata, Languages and Combinatorics
  \textbf{6}(1),  3--23 (2001). \doi{10.25596/JALC-2001-003}

\bibitem{fekete}
Fekete, M.: {\"U}ber die {V}erteilung der {W}urzeln bei gewissen algebraischen
  {G}leichungen mit ganzzahligen {K}oeffizienten. Mathematische Zeitschrift
  \textbf{17},  228--249 (1923). \doi{10.1007/BF01504345}

\bibitem{growth2}
Gawrychowski, P., Krieger, D., Rampersad, N., Shallit, J.: Finding the growth
  rate of a regular or context-free language in polynomial time. In: Proc. DLT.
  LNCS, vol.~5257, pp. 339--358 (2008). \doi{10.1007/978-3-540-85780-8_27}

\bibitem{grigorchuk}
Grigorchuk, R., Mach{\'\i}, A.: An example of an indexed language of
  intermediate growth. Theoretical Computer Science  \textbf{215}(1),  325--327
  (1999). \doi{10.1016/S0304-3975(98)00161-3}

\bibitem{formats2022}
Jacobo~Incl{\'{a}}n, B., Degorre, A., Asarin, E.: Bounded delay timed channel
  coding. In: Proc. {FORMATS}. LNCS, vol. 13465, pp. 65--79 (2022).
  \doi{10.1007/978-3-031-15839-1_4}

\bibitem{kolmoEpsilon}
Kolmogorov, A., Tikhomirov, V.: $\varepsilon$-entropy and
  $\varepsilon$-capacity of sets in function spaces. Uspekhi Matematicheskikh
  Nauk  \textbf{14}(2),  3--86 (1959). \doi{10.1007/978-94-017-2973-4_7}

\bibitem{lambda}
Lancaster, P.: Lambda-matrices and Vibrating Systems. Pergamon Press (1966)

\bibitem{marcus}
Lind, D., Marcus, B.: An Introduction to Symbolic Dynamics and Coding.
  Cambridge University Press (1995)

\bibitem{simply}
Markey, N., Schnoebelen, P.: Symbolic model checking for simply-timed systems.
  In: Proc.~{FORMATS}/{FTRTFT}. LNCS, vol.~3253, pp. 102--117 (2004).
  \doi{10.1007/978-3-540-30206-3_9}

\bibitem{puri}
Puri, A.: Dynamical properties of timed automata. Discrete Event Dynamic
  Systems  \textbf{10}(1-2),  87--113 (2000). \doi{10.1023/A:1008387132377}

\bibitem{simon}
Simon, I.: Factorization forests of finite height. Theoretical Computer Science
   \textbf{72}(1),  65--94 (1990). \doi{10.1016/0304-3975(90)90047-L}

\end{thebibliography}
\newpage
\appendix
\section{Proof details concerning \STA}
\subsection{Proof of \cref{prop:STA}}
\begin{proof}
Let $L$ be the language of an \STA\ and assume that all delays labeling transitions of the \STA\ are multiple of $1/N$ for some $N \in \nat$. 
Let also $K \in \nat$  be an upper bound for all transition's delays.

We have that $$L_{T_1+T_2+K}\subset L_{T_1+K}\cdot L_{T_2+K}$$ (each long run can be split into two shorter ones), 
  thus  
  $$\Upsilon (L_{T_1+T_2+K})\leq \Upsilon (L_{T_1+K})\cdot \Upsilon(L_{T_2+K}),$$ and
  the function $f(T)= \log \Upsilon(L_{T+K})$ is subadditive.
Fekete's lemma \cite{fekete} states that whenever $f$ is subadditive,   $\lim_{T\to\infty} f(T)/T$ exists. In our case, this is the required limit (the growth rate of $L$).

To relate the bandwidth with the growth rate,
we remark that for words $w_1,w_2\in L$ the distance $d(w_1,w_2)$ is always multiple of $1/N$, hence $ d(w_1,w_2)\neq 0 \Leftrightarrow d(w_1,w_2)\geq 1/N$.  Consequently, the notions of  $0$-separated and $\varepsilon$-separated sets coincide for $\varepsilon<1/N$.  This implies that $\capa_\varepsilon(L_T)=\log\Upsilon(L_T)$.

As for $\varepsilon$-entropy, we can use the inequality 
$\capa_{2\varepsilon}(L_T)\leq \ent_{\varepsilon}(L_T)\leq  \capa_{\varepsilon}(L_T)$ from \cite{kolmoEpsilon} to conclude that everything coincides for $\varepsilon<1/2N$:
$$
\ent_\varepsilon(L_T)=\capa_\varepsilon(L_T)=\log\Upsilon(L_T);
$$
dividing by $T$ and taking the limit we obtain the required equality of bandwidths. \qed
\end{proof}

\subsection{Proof of \cref{eq:nu:L}}

\begin{proof}

For any finite set $S$ of timed words over alphabet $\Sigma$, we will prove that
\begin{equation}\label{eq:ups}
 \# \NoZ(S) \leq \Upsilon(S)\leq\# \NoZ(S)(2^{\#\Sigma}-1).   
\end{equation}

The inequality $\#\NoZ(S)\leq \Upsilon(S)$ comes from the fact that for each word $w\in\nu(S)$ one can select one word $\tilde\nu(w)\in S\cap \nu^{-1}(w)$; hence we obtain the set $\tilde\nu(\nu(S))$ which is a $0$-separated subset of $S$ with the cardinality of $\nu(S)$.

For $\Upsilon(S)\leq\#\NoZ(S)(2^{\#\Sigma}-1)$: consider any $0$-separated subset $E$ of $S$. Consider the application $s$ that deletes all the letters at time $0$ from a timed word. Then $\NoZ(s(E))=\NoZ(E)$. Just remark that, because $E$ is $0$-separated, every word in  $\NoZ(s(E))$ only has a single pre-image by $\NoZ$ in $s(E)$ and words from $s(E)$ at most $2^{\#\Sigma}-1$ pre-images in $E$ by $s$. Hence $\#E\leq\#\NoZ(E)(2^{\#\Sigma}-1)\leq\#\NoZ(S)(2^{\#\Sigma}-1)$. Since this is true for any $0$-separated subset $E$ of $S$, we conclude $\Upsilon(S)\leq\#\NoZ(S)(2^{\#\Sigma}-1)$.

The proposition statement is immediate: applying \cref{eq:ups} to $L_T$, we get $\#\Upsilon(\NoZ(L_T))\leq \Upsilon(L_T)\leq\Upsilon(\#\NoZ(L_T))(2^{\#\Sigma}-1)$. Taking the limit of the logarithms divided by $T$, we obtain the desired result. 
\end{proof}

\subsection{Proof of \cref{prop:char}}
\begin{proof}
First, let us define $L_d=\{w\in L(\gra), \length(w)=d\}$ and remark that $\beta(L(\gra))=\lim_{T\to\infty}\frac{\log \sum_{d<T} \#L_d}{T}=\lim_{d\to\infty}\frac{\log \#L_d}{d}$; let us call that number just $\beta$.

Then notice that all the delays in $\gra$ are rational numbers with denominators in a finite set of positive natural numbers $D$. 
Hence, if $m$ is the lcm of $D$, $M_\gra(\zeta^m)$ is a matrix of polynomials in $\zeta$.

Consider the formal matrix series $R(\zeta)=(I-M_\gra(\zeta^m))^{-1}= \sum_{k=0}^\infty M_\gra^k(\zeta^m)$. 
Its elements are $R_{ij}(\zeta)=\sum_{n\in \nat} c_{ijn}\zeta^n$, where $c_{ijn}$ is the number of runs from $q_i$ to $q_j$ with duration $d=\frac{n}{m}$. Remark that $\max_{i,j}c_{ijn}\leq\#L_d\leq\sum_{i,j}c_{ijn}$.

Take some positive $a<|z_0|$, then $\det(I-M_\gra(a))\neq 0$, 
the matrix  $(I-M_\gra(a))^{-1}$ is well-defined, and thus, 
for $\alpha = a^{-m}$, $\sum_n c_{ijn}\alpha^n$ converges and so does $\#L_d$, hence, for $n$ large enough, $\#L_d\alpha^n<1$.
This implies $\#L_d<\alpha^{-md}$ and thus $\beta$, the growth rate wrt $d$, verifies $\beta\leq -m\log \alpha = -\log a$. Since $a$ can be arbitrary close to $z_0$ we conclude that  $\beta \leq -\log|z_0|$.

On the other hand, suppose for contradiction that $\beta < -\log|z_0|-b$ for some positive $b$. Since $\det(I-M_\gra(z_0))=0$, there exists a non-zero vector $u$ such that $M_\gra(z_0)u=u$, thus the series $\sum_k M^k_\gra(z_0)u=\sum_k u$ diverges, hence some of its matrix elements $R_{ij}(z_0)$ diverge. But $c_{ijn}\leq \#L_d \leq K 2^{\beta d}<z_0^{-d} 2^{-bd}$ (for some positive constant $K$ and $d$ large enough) and thus  the general term of $R_{ij}(z_0)$ is smaller than $2^{-bd}$ implying that $R_{ij}(z_0)$ must also converge, which is contradictory.\qed
\end{proof}

\section{Preparatory results}

\subsection{Puri's reachability method}
We recall here some basics on reachability in timed automata, mostly due to \cite{puri} and slightly reformulated in \cite{entroJourn,3classes}. 
We will use a specific coordinate system on clock regions.
\begin{con}[Canonical barycentric coordinates]\label{con:canon} Given a $k-1$-di\-men\-sional region $r$ we compare its vertices as follows:  $v < v'$ whenever $v\cdot\1<v'\cdot\1$ (it is easy to see that this is a total order). 
We enumerate all the vertices in the ascending order: $v_1<v_2<\cdots <v_k$.
For any clock vector $x\in \bar r$ we define its \emph{barycentric coordinates} $\bla(x)=(\lambda_1,\dots,\lambda_k)$ by 
conditions
$$
\sum_{i=1}^k \lambda_i v_i=x;\quad \sum_{i=1}^k \lambda_i =1;\quad  0 \leq \lambda_i \leq 1. 
$$
\end{con}

Given a path  $q\trans{\pi}q'$, and  faces $r$ and $r'$ in $S(q)$ and $S(q')$ respectively, we aim to characterize  when $x\in r$ can go to $x'\in r'$ through $\bar\pi$ (the closure of the path). 
{That is, there exists a run $\rho$ associated with $\bar\pi$.}
The main tool of Puri's method is the \emph{orbit graph} summarizing reachability between vertices regions.

\begin{definition}The \emph{orbit graph} $\gamma(\pi)$ is a bipartite graph from vertices of $r$ to those of $r'$. It has an edge $v\to v'$ whenever $v$ can go to $v'$ through $\bar\pi$. The \emph{orbit matrix} $M(\pi)$ is the adjacency matrix of  $\gamma(\pi)$.
\end{definition}

\begin{theorem}[\cite{puri}]
Given a path $q\trans{\pi}q'$, a state $(q',x')$ (with $x'$ in an $n'$-dimensional region face) is reachable from a state $(q,x)$ (with $x$ in an $n$-dimensional region face) through $\bar\pi$ iff there exists an $n \times n'$ 
stochastic matrix\footnote{That is, $P\cdot \1_n = \1_{n'}$.} $P$ satisfying $\bla(x)P=\bla(x')$ and $P \leq M(\pi)$. 
\end{theorem}
We call such a $P$ a reachability matrix from $x$ to $x'$ along $\bar\pi$. 
We will enrich $P$ with duration information as follows.
If for all $i$ and $j$, such that $P_{ij}>0$, there is at most a single possible run $\rho_{ij}$ from $v_i$ to $v'_j$ through $\bar\pi$, then we define the timing matrix $D^P$, such that 
$$
D^P_{ij}=\begin{cases}
    \length (\rho_{ij})&\text{if } P_{ij}>0;\\
    0 & \text{otherwise}.
\end{cases}
$$

 \begin{proposition}\label{prop:puridur}
 Let $P$ be a reachability matrix from $x$ to $x'$ along $\bar\pi$, with a well-defined timing matrix $D^P$.  
 Then $x$ can go to $x'$ in time  $d=\sum_{i,j}{\lambda_{i}(x)P_{ij}D^P_{ij}}$.
 \end{proposition}

 \begin{proof}[sketch]
 First, remark that, for each $i$, there is a run $\rho_i$ along $\pi$ from $v_i$ to $\sum_j P_{ij} v'_j$ that lasts $\sum_j P_{ij}D_{ij}^P$ (by convex combination of runs from $v_i$ to $v'_j$). Then, we make a convex combination of the $\rho_i$ according to weights in $\bla(x)$. Since $x=\sum_i \lambda_i(x) v_i$, we obtain a run from $x$ to $\sum_i\sum_j P_{ij} v'_j = \sum_j(\sum_i P_{ij}) v_j = x'$ that lasts $\sum_{ij} \lambda_i(x)P_{ij}D_{ij}^P$. \qed
 \end{proof}

\subsection{On cycles in meager timed automata}
Through the next two sections, we consider a cyclic run (along some path $\pi$)
$$
(q,x)= (q_0,x_0) \to (q_1,x_1) \to \cdots \to (q_k,x_k)= (q,x)
$$
in a meager \RTA. Let $r_i = \face(x_i)$ be the smallest face containing $x_i$ and $P_i$ the reachability matrix from $x_{i-1}$ to ${x_i}$. By definition of $\face(x_i)$, all components of $\bla(x_i)$ are positive and
$\bla(x_{i-1})P_i=\bla(x_i)$.
Through a series of lemmas, we will characterize matrices $P_i$ and also the transition delays.

\begin{lemma}\label{lem:self}
    Let $\pi$ be the path of the cyclic run, and take $\pi^j$ its power with idempotent  $\gamma(\pi^j)$. Then the orbit graph  $\gamma(\pi^j)$ has self-loops on all vertices of the face  $r$ of $x$.
\end{lemma}

\begin{proof}
    Suppose this is not the case for some vertex $u$. 
    Take $U$ the set of vertices $v$ such that there is a path from $v$ to $u$ in $\gamma(\pi)$ (and necessarily in $\gamma(\pi^j)$). Let $\lambda_U=\sum_{v \in U}{\lambda_v}$. Then through $\pi$, $\lambda_U(x)$ diminishes which means at least one of the $\lambda_v(x)$ does as well. But then it is impossible to get the same coordinates after the cycle, which means it is impossible to return to $x$. \qed
\end{proof}

\begin{lemma}
All faces $r_i$ have the same dimension. Hence all $P_i$ are square matrices.    
\end{lemma}
\begin{proof}
Suppose this is not the case, and the run contains two points $x$ and $x'$ in faces $r$ and $r'$ of dimensions $n$ and $n'$, with $n<n'$ (this is w.l.o.g. as the cyclic run can be taken from $x'$). We denote the part of the cycle from $x$ to $x'$ by $\pi^1$ and from $x'$ to $x$ by $\pi^2$, and $\pi^1\pi^2$ by $\pi$. Denote $P^1,P^2,P$ the corresponding reachability matrices and $g^1,g^2,g$ the orbit graphs without edges whose coefficient is null in $P^1,P^2,P$.

    Because $n<n'$, there has to be at least one vertex $u$ of $r$ having (at least) two outgoing arrows in $g^1$ to vertices $u'_1\neq u'_2$ of $r'$, which have themselves edges to $v_1$ and $v_2$ (in $r$) along $g^2$ (they might be the same vertex).
    In the region $r$, we consider $U$ the set of vertices $w$ such that there is a path of edges in $g$ going from $w$ to $u$.
    Then both $v_1$ and $v_2$ belong to $U$ 
    because otherwise, from the previous lemma, $u \in U$ and $\lambda_U(x)$ would decrease (either along $(u,v_1)$ or $(u,v_2)$).
    However, if both $v_1$ and $v_2$ are in $U$, let $l$ and $L$ be the shortest paths from each of them to $u$, then the orbit graph of the path $\pi^{(l+1)(L+1)}$ has two ways of going from $u$ to itself, which contradicts $\aut$ being meager. \qed  
\end{proof}


Let $k$ be the dimension of all the faces $r_i$. As a consequence of this lemma the following holds.

\begin{lemma}\label{lem:bij}
Reachability matrices $P_i$ can be written 
containing only 0 and 1, written in this way they define bijections on $k$ vertices.
\end{lemma}

\begin{proof}
    Take $\pi^j$ as in \cref{lem:self}, then in  $\gamma(\pi^j)$ all vertices have self-loops. By following the only path of each vertex to itself, we get a run that goes back to $x$. This is a cyclic run, which means all interior $x_i$ are of the same dimension, so necessarily no vertex-to-vertex paths coincide on any point. This is in particular true for the first occurrence of $\pi$ which gives the expected result. \qed
\end{proof}

%
The next result shows that the delays between vertices are close to each other.
\begin{lemma}\label{lem:dplus1}
    Let $\xi$ ($u=u_0 \trans{d_1} \dots \trans{d_k}u_{k}=u$) and $\xi'$ ($v=v_0\trans{d'_1} \dots \trans{d'_k}v_{k}=v)$ 
    be cyclic runs 
    of vertices  corresponding to the same path $\bar\pi$ {of transitions $\delta_i$} in $\bar\aut$ ($u_i$ and $v_i$ are vertices of $r_i$). Then the following holds.
    \begin{enumerate}
        \item For every $i$, assuming w.l.o.g.~that $u_{i-1} \leq v_{i-1}$, either $d_i=d'_i$ and then $u_i \leq v_i$ or $d_i=d'_i+1$ and then $u_i \geq v_i$.
        \item Both paths have the same duration: $\sum{d_i}=\sum {d'_i}$.
    \end{enumerate}
\end{lemma}

\begin{proof}Note that $\aut$ being meager, these paths are the only ones cycling from $u$ to itself and from $v$ to itself. Hence, if $u=v$ there is nothing to prove.

We consider the case where $u \neq v$. 
For every $i$, $u_i \neq v_i$, otherwise two different paths are going from $u_i$ to itself through the same shifted version of $\pi$ (one passing through $u$, the other through $v$) which is in contradiction with the fact that $\aut$ is meager.

Moreover, one cannot have at the same time $v_i$ reachable from $u_{i-1}$ and $u_i$ reachable from $v_{i-1}$ through $\delta_i$. If this were the case, through the same shifted version of $\pi$ we can go from $u_i$ and $v_i$ to themselves and each other. By repeating this cycle twice, we get two ways of going from $u_i$ (and $v_i$) to itself, which is again impossible in a meager automaton.

We now fix $i$, and we will look at the difference between $d_i$ and $d'_i$. We can consider w.l.o.g.~$u_{i-1}<v_{i-1}$. These being two vertices of the same region $r_{i-1}$, we can partition the set of all clocks as $X=X_{eq} \cup X_{neq}$ such that for all $c \in X_{eq}$, ${(u_{i-1})}_c={(v_{i-1})}_c$ (with $u_c$ being the value of clock $c$ in the vertex $u$) and for all $c \in X_{neq}$, ${(u_{i-1}})_c+1={(v_{i-1})}_c$.

Consider the case where $d_i<d'_i$. Since these timings are integers, this means $d_i+1 \leq d'_i$. Given $c \in X_{neq}$, we have ${(u_{i-1})}_c+d_i + 2 \leq {(v_{i-1})}_c+d'_i$, hence ${(u_{i-1})}+d_i$ and ${(v_{i-1})}+d'_i$ are not be in the same closed region. But $u_i$ and $v_i$ must be in the same closed region, so necessarily $c \in \reset_{\delta_i}$.
We also have that $u_{i-1}+d_i \models g(\delta_i)$ and $v_{i-1}+d'_i \models g(\delta_i)$. Since $u_{i-1}+d_i < v_{i-1}+d_i \leq u_{i-1}+d'_i < v_{i-1}+d'_i$, by convexity (on each coordinate), all these 4 clock vectors satisfy the guard of $\delta_i$. Because clocks that are different in $u_{i-1}$ and $v_{i-1}$ are reset,  $(u_{i-1}+d_i)[\reset_{\delta_i}] = (v_{i-1}+d_i)[\reset_{\delta_i}] = u_i$ and $(u_{i-1}+d'_i)[\reset_{\delta_i}]=(v_{i-1}+d'_i)[\reset_{\delta_i}] = v_i$; this means that $v_i$ is reachable from $u_{i-1}$ and $u_i$ from $v_{i-1}$ through $\delta_i$, which is a contradiction.

So necessarily $d_i \geq d'_i$.

Consider the case where $d_i \geq d'_i+2$. By similar reasoning as before, we get that for $c \in X_{eq}$, $c \in \reset_{\delta_i}$. We also get $v_{i-1} + d'_i \leq u_{i-1} + d'_i + 1 \leq v_{i-1} + d_i - 1 \leq u_{i-1} + d_i$ with all these clock vectors satisfying the guards of $\delta_i$. By taking timings $d'_i+1$ for $u_{i-1}$ and $d_i-1$ for $v_{i-1}$ we reach $v_i$ and $u_i$ respectively, which leads to the same contradiction as before.

So the only possible options are $d_i = d'_i$ or $d_i = d'_i + 1$.

It is easy to notice that whenever $d_i = d'_i$ then $u_i<v_i$ and whenever $d_i = d'_i + 1$ then $u_i>v_i$ (equality of delays preserves ordering, while inequality swaps it). 

Now consider the whole cycle again. Observe that on odd-numbered swaps, we have $d=d'+1$ and, on even-numbered ones, $d+1=d'$. Hence in any factor with an even number of swaps, the duration of both considered $\STA$ paths is the same, but this is necessarily the case for the whole cycle (since $u$ and $v$ cycle back to themselves through the two considered cycles of the \STA). So both cycles have the same duration. \qed
\end{proof}

\begin{corollary}\label{cor:standard} All matrices $P_i$ have one of the forms
$$
P^{kl}=
\left(
\begin{array}{c|c}
0_{l\times(k-l)} & I_l\\ \hline
I_{k-l} & 0_{(k-l)\times l}
\end{array}
\right).
$$
The form  $P^{kl}$ of the transition matrix corresponds to the timing matrix $D^P$ of the form
$$
D^{kld}=
\left(
\begin{array}{c|c}
0_{l\times(k-l)} & (d+1)\cdot I_l\\ \hline
d\cdot I_{k-l} & 0_{l\times(k-l)}
\end{array}
\right).
$$
\end{corollary}
\begin{proof}
    Consider a transition, on $k$-dimensioned faces, where $l$ ``slow'' vertices make it in $d+1$ time units, while the others (``fast '') make it in $d$ time units. It follows from \cref{lem:dplus1}, that the slow vertices are   $v_1\dots v_l$  and they go to $v'_{k-l+1}\dots v'_k$ (in the same order), while the fast ones are $v_{l+1}\dots v_k$ that go to $v'_1\dots v'_{k-l}$. In matrix form, this corresponds to the statement to prove. \qed
\end{proof}

\begin{corollary}
    Let $\rho$ and $\rho'$ be two cyclic runs on $x$ and $x'$ such that $path(\rho)=path(\rho')=\pi$, let $w=word(\rho)$ and $w'=word(\rho')$. Then $w$ and $w'$ have the same duration. 
\end{corollary}

\begin{proof}
    Consider $\pi^j$ a power of the path such that its orbit graph is idempotent. In this orbit graph, all vertices have self-loops, from \cref{lem:dplus1}, going from each of these vertices to itself via $\pi^j$ is done in the same amount of time, which means that a cycle through $\pi^j$ is done in the same amount of time $t$ in any point of the face. Along the orbit graph of $\pi$, the same vertices either have self-loops, or they are in a permutation with other vertices having the same weight in $x$ and $x'$. In any case, the duration of a cyclic run is necessarily $t/j$. \qed
\end{proof}

\subsection{Relating cycles with barycenters}
\begin{lemma}
    Let $\rho$ be a cyclic run of path $\pi$ on $x$ (of dimension $k-1$), then there is a cyclic run of the same path on barycenters $\alpha(x)$. Moreover, the reachability matrices are the same on both runs.
\end{lemma}

\begin{proof}
    Consider the $P_i$ from the previous section for the run on $x$. Applying the same matrices we get a run of the same path starting on $\alpha(x)$. Moreover, as $\bla(\alpha(x_0))=(\frac{1}{k},\dots,\frac{1}{k})$ and each $P_i$ is a bijection on vertices, all the $\bla(\alpha(x_i))$ will have exactly the same form, and thus be barycenters of regions $r_i$. The run of barycenters will lead back to $\alpha(x)$ which concludes the proof. \qed
\end{proof}

Consider now delays in the barycentric cycle.
\begin{lemma}\label{lem:timeandkind}
    In a barycentric cycle, whenever the duration matrix for transition $i$ has the form $D^{kld}$ from \cref{cor:standard}, the transition $\alpha(x_{i-1})\to \alpha(x_i)$ takes $d+l/k$ time.
\end{lemma}
\begin{proof} We notice, that due to meagerness, all delays in a barycentric cycle are uniquely determined. Next, we apply  \cref{prop:puridur} to the matrix $D^{kld}$ and the barycenter with barycentric coordinates $(\frac{1}{k},\dots,\frac{1}{k})$. \qed
\end{proof}
We will use the previous lemma in the reverse direction: knowing that the barycenter takes a transition in $d+l/k$ time, we will deduce $l,k,d$, then matrices $P^{kl}$ and $D^{kld}$, and finally the successors and timings in the concrete run of $x_i$

\begin{lemma}\label{lem:tau:sigma}
For any  transition $x\xrightarrow{\bar\delta, t} x'$ on a cyclic run, let
\begin{itemize}
\item $f$ and $f'$ be the region faces containing $x$ and $x'$,  $k-1$ their dimension,  $v_1,\dots,v_{k}$ and $v'_1,\dots,v'_{k}$ their (ordered as in \cref{con:canon}) vertices; 
\item  $\bla$ and $\bla'$ the barycentric coordinates of $x$ and $x'$ with respect to vertices;
\item $\alpha(x)\xrightarrow{\bar\delta,t_0}\alpha(x')$ the corresponding transition between barycenters of faces $f$ and $f'$.
\end{itemize}
Then  $t= \tau_k(\bla,t_0)$  and $\bla' = \sigma_k(\bla,t_0)$,
where functions $\tau$ and $\sigma$ are always the same (for all states and all automata). Moreover,  $\tau_k(\bla,t_0)=0 \Leftrightarrow t_0=0$ and $\sigma_k(\bla,0)=\bla$.
\end{lemma}
\begin{proof}
We first write  $t_0$ as $d+l/k$ (with $0\leq l<k$), and conclude by \cref{lem:timeandkind} that reachability and delay matrices are $P^{kl}$ and $D^{kld}$.
Hence $x=\sum_{i=1}^k\lambda_i v_i$ would make the transition in $d+\sum_{i=1}^j\lambda_i$ time units and arrive to
$$
x'= \sum_{i=j+1}^k\lambda_i v'_{i-j}
+\sum_{i=1}^j\lambda_i v'_{k-j+i}.
$$
Thus the transition time and the barycentric coordinates of the target are uniquely determined. To be formal,
$$
t=\tau_k(\bla,d+j/k)\triangleq d+\sum_{i=1}^j\lambda_i,\quad \bla'=\sigma_k(\bla,d+j/k)\triangleq\lambda_{j+1}\dots\lambda_k\lambda_1\dots\lambda_j. 
$$
The statements concerning $t_0=0$ are now immediate. \qed
\end{proof}

\subsection{Corollary from \cite[full version]{3classes}}

\begin{definition}
    A number $b\in\real$ is called a \emph{cycle growth bound} of an \RTA\ $\aut$ whenever there exists a constant $C$ such that for any $q\in Q, x\in \bar{S}(q), T>0$, the language $L_T^{\bar{\aut}}((q,x),(q,x))$ admits a $0$-net of cardinality $C2^{bT}$. 
\end{definition}

\begin{proposition}\label{prop:cgb}
If $b\in\real$ is a \emph{cycle growth bound} of a meager  \RTA\ $\aut$, then (for any $\varepsilon>0$) 
$$\bandh_\varepsilon(\aut)\leq \bandc_\varepsilon(\aut)\leq b.$$
\end{proposition}
The quite technical (and based on Simon factorization) proof of this proposition follows verbatim that of Lemmas 68--77 and Prop.~18   from  \cite[full version]{3classes}. 

\subsection{Relating open and closed semantics}

\begin{lemma}\label{lem:closure} For any $\varepsilon>0$ and any word $w\in L(\bar\aut)$, there exists $w'\in L(\aut)$ such that $d(w,w')<\varepsilon$.\end{lemma}
\begin{proof}[sketch]
It was proved in \cite{3classes} that for any path $\pi$ we can define a homothety $h_\lambda$ with a factor $\lambda\in (0,1)$ (its center is fixed independently of $\lambda$ as any interior point of the polyhedron of correct timings) such that, if we fix any run $\rho$ along $\pi$, having a timing vector $\vec t$ satisfying the closed version of the constraints of $\pi$, then $h_\lambda(\vec{t})$ satisfies the open version of the constraints of $\pi$ (hence, changing the timing of $\rho$ from $\vec{t}$ to $h_\lambda(\vec{t})$, with $\lambda$ close enough to $1$, we 
obtain a run $\rho'$ through $\pi$ in $\aut$, such that $d(\Word(\rho),\Word(\rho'))<\varepsilon$). To lift the previous result to the whole automaton, it remains to remark that the paths of $\aut$ and $\bar\aut$ are the same since the two automata have the same underlying graph. \qed
\end{proof}

Remark that this result only holds thanks to region-splitting: indeed, the result relies on Puri's characterization of reachability between regions. In particular, it is well known that general TA may have non-realizable paths that become realizable after closing guards (in this case, no homothety can make their languages similar).

\section{Proof of main \cref{prop:main}}


We have proved on \cpageref{prop:main} that
\begin{equation}
\bandc_\varepsilon(\aut)\geq\beta(\alpha(\aut)); \quad 
\bandh_\varepsilon(\aut)\geq \beta(\alpha(\aut)).
\end{equation}

The proof (below) of the opposite inequality is based on \cref{prop:cgb}.

\subsection{Bounding the cycle growth}
It remains to prove that the growth rate of $\gra=\alpha(\aut)$ is a cycle growth bound for $\aut$. To that aim, we concentrate on cycles in $\aut$.

\begin{con}
  Given a timed word $w$  accepted by a (unique)  cyclic run $\rho$ from $(q,x)$ to $(q,x)$, let us consider its abstraction $\alpha(\rho)$. In turn, it accepts some $w'\in  L^{\gra}((q,\alpha(x)),(q,\alpha(x))$. We denote $w'=\alpha_{qx}(w)$.  
\end{con}
Let us summarize some of its properties, immediate from the preceding lemmas.
 \begin{lemma}\label{lem:alpha:words}
 Whenever $w'=\alpha_{qx}(w)$, for $w = (a_1,t_1)\dots(a_n,t_n)$,
 \begin{itemize}
\item $w'$ has the form $w' = (a_1,t'_1)\dots(a_n,t'_n)$; 
\item its duration is the same as that of $w$, i.e. $t'_n=t_n$; 
\item admitting that the cycle for $w$ passes through the states with barycentric coordinates 
$\bla^0, \bla^1,\dots, \bla^n=\bla^0 $, we have that
$$t_i-t_{i-1}=\tau(\bla^{i-1},t'_i-t'_{i-1}),\quad \bla^i=\sigma(\bla^{i-1},t'_i-t'_{i-1}). $$
 \end{itemize}
 \end{lemma}

\begin{lemma}\label{lem:sep2sep}
Let  $w^1,w^2\in L_T^{\bar{\aut}}((q,x),(q,x))$ be two timed words such that \\ $d(w^1,w^2)>0$ (necessarily they are recognized along different cyclic paths), and ${w^i}'=\alpha_{qx}(w^i), i=1,2$. Then $d({w^1}',{w^2}')>0$.
\end{lemma}
\begin{proof}
\eugene{peut-on le dire plus simplement avec le lemme précédent}
For contradiction, suppose that  $d({w^1}',{w^2}')=0$. Then by \cref{eq:nu}, $\NoZ({w^1}')=\NoZ({w^2}')$, let it be:
$$
\NoZ({w^1}')=\NoZ({w^2}')= (A_1,t'_1),(A_2,t'_2),\dots, (A_n,t'_n)
$$
with $A_j\in 2^\Sigma$. Let $A_0$ be the set of events happening in ${w^1}'$ and ${w^2}'$ at time $0$ (necessary it is the same set, eventually empty).
By construction of ${w^i}'$, the words $w^i$ have the same form,  with eventually different timing (in particular, the last statement of \cref{lem:tau:sigma} guarantees that $w^i$ and ${w^i}'$ have 0-time transitions at the same positions), thus
$$
\NoZ({w^i})=(A_1,t^i_1),(A_2,t^i_2),\dots, (A_n,t^i_n), \text{ with } i=1,2,
$$
and $w^i$ also have the set of events $A_0$ at the beginning.

Let $\bla_0$ be the barycentric coordinates of $x$. The run of $\aut$ from $(q,x)$ over $w^i$ visits the points with barycentric coordinates:
\begin{itemize}
    \item $\bla_0$ in the beginning;
    \item $\sigma(\bla_0,0)=\bla_0$ again after reading every letter of $A_0$, by \cref{lem:tau:sigma};
    \item $\bla_1=\sigma(\bla_0,t'_1)$ after the next transition;
    \item  and again $\bla_1$ after reading  each letter of $A_1$;
    \item  $\bla_2=\sigma(\bla_1,t'_2-t'_1)$ etc.
\end{itemize}
We see that runs on both words $w^1$ and $w^2$ follow the same sequence of barycentric coordinates. Applying again \cref{lem:tau:sigma} we conclude that they also have the same timing since the transition from $\bla_j$ to $\bla_{j+1}$ always 
takes $\tau(\bla_j,t'_{j+1}-t'_j)$ time. From that we deduce  that $\NoZ({w^1})=\NoZ({w^2})$, and both $w^1$ and $w^2$ have the same set $A_0$ of events at time $0$, thus 
$d(w^1,w^2)=0$, a contradiction. \qed
\end{proof}

\begin{lemma}\label{lem:sep}
    $\beta(\gra)$ is a cycle growth bound for $\aut$.
\end{lemma}
\begin{proof}
Let $\mathcal{N}$ be a maximal $0$-separated set in $L_T^{\bar{\aut}}((q,x),(q,x))$ (it is also a  minimal $0$-net).
Its image
$\mathcal{N}'=\alpha_{qx}(\mathcal{N})$ by virtue of \cref{lem:sep2sep} is $0$-separated and of the same cardinality. By construction  $\mathcal{N}'\subset L_T^{\gra}((q,\alpha(x)),(q,\alpha(x)))$. Hence $\#\mathcal{N}= \#\mathcal{N'}\leq C2^{\beta T}$, thus $\beta(\gra)$ is a cycle growth bound for $\aut$.  \qed
\end{proof}
\cref{prop:main} follows immediately  from \cref{prop:cgb,lem:sep}. 
 
\end{document}